\documentclass[conference]{IEEEtran}
\IEEEoverridecommandlockouts
\usepackage{cite}
\usepackage{amsmath,amssymb,amsfonts}
\usepackage{algorithmic}
\usepackage{graphicx}
\usepackage{textcomp}
\usepackage{xcolor}
\usepackage{amsthm}
\usepackage{tabstackengine}
\setstackEOL{\cr}
\newtheorem{theorem}{Theorem}
\newtheorem{corollary}{Corollary}

\newtheorem{lemma}[theorem]{Lemma}
\newcommand{\diag}{\operatorname{diag}}
\newcommand{\rank}{\operatorname{rank}}
\renewcommand{\vec}{\operatorname{vec}}
\usepackage{booktabs}
\usepackage[utf8]{inputenc}
\usepackage{lscape}
\usepackage{textcomp}
\usepackage{verbatim}
\usepackage{subcaption}
\usepackage{longtable}

\usepackage{array}
\usepackage{makecell}
\usepackage{graphicx}
\usepackage{hyperref}

\def\BibTeX{{\rm B\kern-.05em{\sc i\kern-.025em b}\kern-.08em
    T\kern-.1667em\lower.7ex\hbox{E}\kern-.125emX}}
    
\newcommand{\Real}{\mathbb{R}}
\newcommand{\Complex}{\mathbb{C}}
\newcommand{\Integer}{\mathbb{Z}}

\newcommand{\cb}{\mathbf{c}}
\newcommand{\gb}{\mathbf{g}}
\newcommand{\xb}{\mathbf{x}}
\newcommand{\zb}{\mathbf{z}}

\newcommand{\omegab}{\boldsymbol{\omega}}

\newcommand{\Cmat}{\mathbf{C}}

\newcommand{\Dc}{\mathcal{D}}
\newcommand{\Lc}{\mathcal{L}}

\newcommand{\krank}{\mathrm{krank}}
    
\begin{document}

\title{Dynamic Tomography Reconstruction by Projection-Domain Separable Modeling  
\thanks{This research was supported in part by Los Alamos National Labs under Subcontract No. 599416/CW13995.}
}

\author{%
\IEEEauthorblockN{Berk Iskender and Yoram Bresler}
\IEEEauthorblockA{%
\textit{Department of ECE and Coordinated Science Lab} \\
\textit{University of Illinois at Urbana-Champaign, IL, USA.}%
}%
\and
\IEEEauthorblockN{Marc L. Klasky}
\IEEEauthorblockA{%
\textit{Los Alamos National Laboratory,} \\ \textit{Los Alamos, NM, USA.}%
}%
}%

\maketitle

\begin{IEEEkeywords}
Dynamic tomography, Partially-separable, Bilinear, Unique recovery.
\end{IEEEkeywords}

\begin{abstract}
In dynamic tomography the object undergoes changes while projections are being acquired sequentially in time. The resulting inconsistent set of projections cannot be used directly to reconstruct an object corresponding to a time instant. Instead, the objective is to reconstruct a spatio-temporal representation of the object, which can be displayed as a movie. We analyze conditions for unique and stable solution of this ill-posed inverse problem, and present a recovery algorithm, validating it experimentally.
We compare our approach to one based on the recently proposed GMLR variation on deep prior for video, demonstrating the advantages of the proposed approach. 
\end{abstract}
\section{Introduction}

\IEEEPARstart{T}{he} dynamic tomography problem addresses the recovery of a time-varying object from projections acquired sequentially at specific time instants. Since the object evolves in time, and too few projections (only one, in the extreme case) are acquired at any time instant, they are
insufficient to reconstruct the object at any time.
The problem arises in the field of medical imaging \cite{10.1093/ehjci/jew044}, imaging of fluid flow processes \cite{scanziani2018situ,o2011dynamic}, and certain microscopic tomography tasks \cite{maire201620}. 

Previous work in this field includes \cite{willis95, willis95_2}, which  provides  theoretical guarantees of unique and stable reconstruction, and an algorithm of reconstruction of spatially-localized temporal objects using an optimal sampling pattern. This approach is limited by its assumptions of  temporal and spatial bandlimits. 

A two-step algorithm \cite{ma11081395} that alternates between estimating the motion field and the object has promising empirical results, but there are no guarantees for unique recovery. Others 
\cite{robinson2003fast, xiong2014linearly} assume the object information a priori, and use the properties of the Radon transform to efficiently recover the motion field.

Work on dynamic MRI \cite{Haldar2010,ma2019dynamic,ma2021dynamic} considers a partially separable object model (PSM) and uses results from the field of low-rank matrix recovery for reconstruction. This framework is applicable to the problem considered in this work. However, the method is not tailored specifically for tomography, and does not provide a clear insight into designing the angular sampling order, or relevant condition numbers.

\textit{Contributions:} We present a new unsupervised dynamic tomographic reconstruction algorithm dubbed ProSep that uses a special bilinear partially-separable model for the projections of a time-varying object.
We consider a specific object-independent angular sampling order for time-sequential sampling of the projections for this model and analyze factors affecting uniqueness and stability of the solution. ProSep does not use any spatial prior for the object, but in numerical experiments shows performance superior to the recently proposed GMLR \cite{Hyder2020} - a deep image prior model for video. We expect that combining a spatial image prior with ProSep will improve its performance even further.

\section{Problem Statement}
In a 2D setting, the goal in dynamic tomography is to reconstruct a time-varying object $f(\xb,t)$, $\xb \in \Real^2$ 
vanishing outside a disc of diameter $D$, 
from its projections  
\begin{align}
    g(s,\theta,t) = \mathcal{R}_\theta \{f(\mathbf{x},t)\}
\end{align}
obtained using the Radon transform operator $\mathcal{R}_\theta$ at angle $\theta$. We consider time-sequential sampling of a single $\theta$ at each time, assuming that the acquisition of a projection is fast enough that the object is 
essentially 
static within the sampling time, but the variation between samples cannot be ignored.
Assuming sampling uniform in time, the acquired data is
\begin{equation}
\label{eq:sampled_proj}
    \{g(s,\theta_p, t_p)\}_{p=0}^{P-1}, \,\, \forall s, t_p = p\Delta_t,
\end{equation}
where $s$ is the offset of the line of integration from the origin (i.e., detector position), and $P$ is the total number of projections (and temporal samples) acquired. The $s$ variable is uniformly sampled to $\{s_j\}_{j=1}^{J}$, and we assume that this sampling is fine enough to not affect the accuracy of the reconstruction, and therefore suppress it in the notation unless relevant. We consider the sequence $\{\theta_p\}_{p=0}^{P-1}$, with $\theta_p \in[0, 2 \pi]$, which we call the angular sampling scheme, to be a free design parameter\footnote{An arbitrary angular sampling scheme is easily implemented in radial MRI, but also in CT systems such as micro-CT and industrial CT where the factor limiting acquisition speed is the x-ray exposure or sensor readout, rather than the rotation of the object or source.}. 
  
Our objective is to reconstruct a time-sequence, a movie of the object $\{f(\xb,t_p)\}_{p=0}^{P-1}$ from the time-sequential projections in \eqref{eq:sampled_proj}. For a static object, $f(\xb)$ can be reconstructed from $g$ using the inverse {R}adon transform $\mathcal{R}^{-1}$, which is well-approximated in practice by the filtered backprojection (FBP) algorithm for $P$ large enough \cite{epstein2003introduction}. However, time-sequentially acquired projections \eqref{eq:sampled_proj}  are inconsistent because different projections correspond to different objects, and a direct reconstruction using $\mathcal{R}^{-1}$ leads to severe artifacts. 
 
We wish to reconstruct the time series of the dynamic object using minimal and verifiable assumptions, and to analyze the effect of various problem parameters, including the  sampling scheme, on the uniqueness and stability of the reconstruction. We use a high benchmark for accuracy:  FBP-based reconstruction of each frame $f(\xb,t_p)$ from the complete set of $P$ projections \emph{taken simultaneously,} i.e., using a total of $P^2$ rather than just $P$ projections.

\section{Partially Separable Model (PSM)}
\subsection{Partially separable model in the object domain}
The representation of a dynamic object $f(\mathbf{x},t)$ by a $K$-th order partially separable model (PSM) is the series expansion
\begin{equation} \label{eq:parsep}
    f(\mathbf{x},t) = \sum_{k=0}^{K-1} f_{k}(\mathbf{x}) \psi_{k}(t),
\end{equation}
which is known to be dense in $\mathcal{L}_2$ \cite{methods_math_phys}, meaning that any finite energy object can be approximated arbitrarily well by such a model of sufficiently high order. Empirically, it is found that even modest values of $K$ provide high accuracy in applications to MR cardiac imaging \cite{Haldar2010,ma2019dynamic,ma2021dynamic}, however we have not found a quantitative analysis of this phenomenon. Our analysis (see Appendix~\ref{app:par-sep_model_for_aff_motion}) shows that for a spatially bandlimited object, a  time-varying affine transformation (i.e, combination of time-varying translation, scaling, and rotation) of bounded magnitude leads to a good approximation by a low order PSM model. Because our benchmark is FBP-based tomographic reconstruction of a static object from $P$ simultaneous projections, which is inherently spatially bandlimited, the above analysis lends support to the use of the PSM with modest $K$. This reduces the number of degrees of freedom in the dynamic object, enabling reconstruction with less data. 

\subsection{Partially separable model in the projection domain}
While the PSM has been used for dynamic imaging in MRI \cite{Haldar2010,ma2019dynamic,ma2021dynamic}, we gain additional insight into its role in tomography by carrying it into the projection domain, and the harmonic representation of the projections.
For real-valued time-varying projections, the circular harmonic expansion is
\begin{equation}
\label{eq:proj_harm_exp}
    g(s,\theta,t) = \sum_{n\in \Integer} h_n(s,t)e^{jn\theta}, \quad h_n = h^*_{-n} , 
    \end{equation}
and we have the following result.

\begin{theorem}
\label{thm:proj_special_par_sep}
Let an object $f(\bold{x},t)$ vanish outside a disk of diameter $D$ and  have the partially-separable representation 
\begin{equation}
\label{eq:f_parsep_w_err}
f(\bold{x},t) = \sum_{k=0}^{K-1} f_k(\bold{x}) \psi_k(t) + \gamma_f(\bold{x},t),
\end{equation}
with  error term  bounded as $||\gamma_f||_2^2 \leq \epsilon_f$, for some $\epsilon_f > 0$. Then the projections admit the representation
\begin{equation}
\label{eq:proj_par_sep_rep_w_err}
g(s,\theta,t) = \sum_{n \in \Integer} h_{n}(s,t) e^{jn\theta} + \gamma_g(s,\theta,t)
\end{equation}
with $||\gamma_g||_2^2 \leq  \pi D\epsilon_f$ and $h_n(s,t)$ represented by the special
partially separable model
\begin{equation} \label{eq:par-sep-special}
    h_n(s,t) = \sum_{k=0}^{K-1} \beta_{n,k}(s) \psi_{k}(t) .
\end{equation}
\end{theorem}
\textbf{Remark:} The PSM \eqref{eq:par-sep-special} is special, in that all the $h_n(s,t)$ are expanded in a common temporal ``basis" $\{\psi_k(t)\}$ - independently of $n$. Moreover, the approximation error of the projections using this special PSM  for the expansion coefficients  is bounded explicitly in terms of corresponding error in the object PSM \eqref{eq:f_parsep_w_err}. This can be used to show that the special PSM \eqref{eq:proj_par_sep_rep_w_err} is dense in the space of $\mathcal{L}_2$  functions $h_n(s,t)$.

The special PSM for the projections compresses their representation from $\approx PD^2$ parameters to $\approx KD^2+ KP$, with $K \ll P$. We introduce further compression to $\approx KD^2+ Kd$, by modeling each temporal function by $d \ll P$ parameters. 

\section{The recovery problem: Analysis}
\subsection{Representing the Sampled Projections}
Substituting  \eqref{eq:par-sep-special} into \eqref{eq:proj_harm_exp} it can be shown that the projections are given by 
\begin{equation}
\label{eq:face_splitting_expansion}
    {\gb}(s) = (\Theta \bullet \Psi) {\beta}(s)
\end{equation}
where ${\gb}(s) \in \mathbb{R}^{P}$, $\Psi \in \mathbb{R}^{P\times (K+1)}$, $\Theta \in \mathbb{C}^{P \times 2N+1}$, and ${\beta}(s) \in \mathbb{R}^{(2N+1)(K+1)}$ are defined as
\begin{align}
&{\gb}(s) = 
    \begingroup 
    \setlength\arraycolsep{2pt}\begin{bmatrix}
    g(s, \theta_1, t_1), &
    g(s, \theta_2, t_2), 
    &
    \ldots &
    g(s, \theta_P, t_P)
    \end{bmatrix}^T
    \endgroup, \nonumber\\
    &\Psi = \setlength\arraycolsep{1pt} \begin{bmatrix}
    \psi_0(t_1) & \ldots & \psi_K(t_1) \\
    \vdots & \vdots & \vdots \\
    \psi_0(t_P) & \ldots & \psi_K(t_P)
    \end{bmatrix},  \quad
    \widetilde{\Theta}= \begin{bmatrix}
    1 & e^{j\theta_1} & \ldots & e^{j2N\theta_1} \\
    \ldots & \ldots & \ldots & \ldots \\
    1 & e^{j\theta_P} & \ldots & e^{j2N\theta_P} \\
    \end{bmatrix}
    \nonumber\\
    &\Theta = \text{diag}([e^{-jN\theta_1}\,\, \ldots\,\, e^{-jN\theta_P}]) \widetilde{\Theta} \label{eq:Theta} \\
    &{\beta}(s) = 
    \begingroup 
    \setlength\arraycolsep{1pt}\begin{bmatrix}
    \beta_{-N,0}(s) &
    \ldots &
    \beta_{-N,K}(s) &
    \ldots &
    \beta_{N,0}(s) &
    \ldots &
    \beta_{N,K}(s)
    \end{bmatrix}^T
    \endgroup. \nonumber
\end{align}
and $\bullet$ denotes the Face-splitting product \cite{slyusar1999family},
where the $p$-th row of $\Theta \bullet \Psi$ is given by the Kronecker product of the respective rows, $\Theta_p \otimes \Psi_p$.

Next, we leverage the $\pi$-symmetry of flipped projections, $g(-s, \theta) = g(s, \theta+\pi) $, to double the number of equations for the same number of unknowns. Replacing $\theta_p$ by $\theta_p+\pi$ in \eqref{eq:Theta} corresponds to post-multiplication of $\Theta$ by a diagonal matrix with entries $e^{jn\pi}, n = -N, \ldots N$. Combining the resulting equations results in expanding model \eqref{eq:face_splitting_expansion}-\eqref{eq:Theta} to
\begin{align}
    \widehat{\gb}(s) &= 
    [\gb^T(s) , \gb^T(-s)]^T, 
    \widehat{\Psi} = 
    [\Psi^T, 
    \Psi^T ]^T, 
    \widehat{\Theta} = 
        [ \Theta^T ,
        \bar{\Theta}^T ]^T,
 \nonumber 
    \\
    {\widehat{\gb}}(s) &= (\widehat{\Theta} \bullet \widehat{\Psi}) {\beta}(s), 
    \quad \bar{\Theta} = \Theta 
    \,\, 
    \text{diag} \left\{ (-1)^n\right\}_{n=-N}^{N} . \label{eq:barTheta}
\end{align}

\subsection{Low Dimensional Model for  Temporal Functions $\Psi$}
\label{sec:low_dim_temp_fct}
Suppose that a given fixed number of $K$ sampled temporal functions $\psi_k \in \mathbb{R}^P$ reside in a $d$-dimensional subspace where $d>K$, but $d<<P$. 
They can be expressed
as $\psi_k = U z_k$ where $U \in \mathbb{R}^{P \times d}$ is a fixed interpolator and $z_k \in \Real^d$.  
Without loss of generality (wlog) we assume that $U^T U = I_d$, 
i.e.
$U$ has orthonormal columns.  Then, the forward model 
\eqref{eq:barTheta} becomes
\begin{equation}
\label{eq:low_d_forward_model}
    {\widehat{\gb}}(s) = (\widehat{\Theta} \bullet (\widehat{U}Z)) {\beta}(s), \quad \widehat{U} = [U^T, U^T]^T.
\end{equation}

\subsection{From Recovered Representation to Reconstructed Object}
Once \eqref{eq:low_d_forward_model} is solved for the unknowns $Z$ and $\beta(s) \, \forall s$, the reconstructed time series of the object is obtained by first forming the 
temporal functions $\Psi = UZ$ and then obtaining the harmonic expansion coefficients using \eqref{eq:par-sep-special}. 
This allows to obtain the projections at all $\theta \in [0,2\pi]$ for each $t$ using \eqref{eq:proj_harm_exp}. Finally, reconstructions at $\{t_p\}_{p=0}^{P-1}$ of the estimated projections are performed using FBP.

\subsection{Bilinear problem and necessary conditions for recovery}
\label{sec:bilinear_prob_and_nec_cond}
The problem of recovering $Z$ and ${\beta}$ from the sampled projections in \eqref{eq:sampled_proj} using \eqref{eq:low_d_forward_model} is the bilinear problem
\begin{align}
    &\forall s\, \text{find} \, {\beta}(s) \in \mathcal{B}, Z \in \mathcal{Z} \nonumber \\
    &\text{s.t.} \,\, (\widehat{\Theta} \bullet (\widehat{U}Z)){\beta}(s) = \widehat{\gb}(s),
    \label{eq:bilinear_eq}
\end{align}
where $\mathcal{B}$ and $\mathcal{Z}$ are appropriate constraint sets. We wish to study conditions for uniqueness of the solution to problem \eqref{eq:bilinear_eq} and its stability to perturbations in the data and the model. Since
$\beta(s)$ and $Z$ appear in product form, there is an inherent scaling ambiguity \cite{li2016identifiability} in \eqref{eq:bilinear_eq}. To remove it, we impose wlog the constraint that $Z^TZ =I$, or equivalently, that $\mathcal{Z}$ is the $d$-dimensional Stiefel manifold in $\Real^P$.

Next, we investigate conditions for uniqueness and stability of the solution to the bilinear problem. 
As shown in \cite{li2016identifiability}, a necessary condition is that when one of the two variables is fixed to a valid solution, the solution for the other is unique. This motivates the study of the linear inverse problems defined by fixing in turn one of the two variables $\beta(s) \,\,\forall s$ or $Z$ in \eqref{eq:bilinear_eq}.

\smallskip
\textbf{(i) Model Linear in ${\beta}(s)$}.
 Defining
\[ 
     L_1(Z) \triangleq \widehat{\Theta}\bullet \widehat{\Psi} = \widehat{\Theta} \bullet (\widehat{U}Z) 
\]
 yields
 \begin{equation} \label{eq:L1beta}
 \qquad  \widehat{g}(s) = L_1(Z) {\beta}(s)
 \end{equation}
with the $i$-th row of $L_1(Z)$  given by 
\begin{equation} \label{eq:L1Zi}
    L_1(Z)_{i:} = \mathbf{z}^T (A_i \otimes \widehat{U}_{i:}^T )
\end{equation}
where $\mathbf{z} \triangleq \vec (Z) \in \Real^d(K+1)$ and
\begin{equation}
\label{eq:A_i_def}
A_i \triangleq \widehat{\Theta}_{i:}^T \otimes I_{K+1} \in  \mathbb{C}^{(K+1)\times(2N+1)(K+1)} .
\end{equation}
For given $\widehat{\Psi}$ or $Z$, \eqref{eq:L1beta} is a linear inverse problem in ${\beta}(s)$. 

We have the following result for the full-rankness of $L_1$.

\begin{theorem}
\label{thm:L1_full_rank}
Let $L_1(Z) = \widehat{\Theta} \bullet \widehat{\Psi}$, where $\widehat{\Theta}$ and $\widehat{\Psi}$ are defined in \eqref{eq:barTheta}. Suppose that
 ${\Psi}$ with orthonormal columns is random drawn from an absolutely continuous probability distribution on the Stiefel manifold $V_{(K+1)}\big(\mathbb{R}^P\big)$.
If $2P \geq (K+1)(2N+1)$ and the $P$ view angles $\theta_i \in [0, \pi], i=1, \ldots, P$ are distinct, then $L_1$
has full column rank w.p.~1. 
\end{theorem}
Theorem~\ref{thm:L1_full_rank}  implies that if $2P \geq (K+1)(2N+1)$, then for almost all (i.e., generic) $\Psi$ with orthonormal columns, $L_1$ will have full column rank. In practice, we implement the low-dimensional model ${\Psi} = {U}Z$ of Section~\ref{sec:low_dim_temp_fct} with a fixed $U$ satisfying $U^TU = I_d$ and $Z \in \mathbb{R}^{d \times (K+1)}$ with $d \geq K+1$ and $Z^TZ = I_{K+1}$. It then follows that $\Psi^T\Psi= I_{K+1}$. Numerical experiments in Section~\ref{sec:experiments}
for such structured  $\Psi$  resulted not only in full-rank $L_1(Z)$, but in fact in low condition number for appropriately chosen view-angle sampling scheme.

This provides confidence that this necessary full-rank condition for the uniqueness of the solution of \eqref{eq:bilinear_eq} will be satisfied.

\textbf{(ii) Model Linear in $Z$.} Here we introduce explicitly the discretization of the detector positions $s$ to $J$ values. For fixed $\beta$ it is shown in Appendix~\ref{app:C} that \eqref{eq:bilinear_eq} reduces to a linear inverse problem in $\zb \triangleq \text{vec}(Z)$
\begin{equation} \label{eq:Linz}
    \widehat{\gb} = \bold{L}_2({\beta}) \zb
\end{equation}
where vector $\widehat{\gb} \in \mathbb{R}^{2JP}$ includes the stacked elements of ${g}(s)_i$, $s= \{s_j\}_{j=1}^{J}$, $i=1, \ldots , 2P$ and  $\bold{L}_2({\beta}) \in \mathbb{R}^{2JP \times d(K+1)}$,  
with $i$-th row given by 
\begin{equation}
\label{eq:L_2_ith_row_def}
    L_2({\boldsymbol{\beta}(s)})^{(i)} \triangleq {\beta}(s)^T A_i^T ( I_{K+1} \otimes \widehat{U}_{i:} ) 
\end{equation}

Problem \eqref{eq:Linz} will have a unique solution for fixed $\boldsymbol{\beta}$ if and only if  $\bold{L}_2(\boldsymbol{\beta})$ has full column rank. 
Furthermore, the stability of the solution is governed by the condition number $\kappa(\bold{L}_2(\boldsymbol{\beta}))$.

\begin{theorem}
\label{thm:L2_cond_num}
Let $\Gamma \triangleq \sum_{j=1}^{J}  \boldsymbol{\beta}(s_j)\boldsymbol{\beta}(s_j)^T \succ 0$.
Then, $\bold{L}_2(\boldsymbol{\beta})$ is full column rank, and
\begin{equation} \label{eq:L2condbound}
    \kappa(\bold{L}_2(\boldsymbol{\beta})) \leq \sqrt{\kappa(\Gamma)}.
\end{equation}
\end{theorem}
Clearly, {$\Gamma \succ 0$, i.e.,} $\kappa(\Gamma) < \infty$ establishes full column rank of $\bold{L}_2(\boldsymbol{\beta})$ and uniqueness of the solution to \eqref{eq:Linz} for fixed $\boldsymbol{\beta}$.
This requires $J \geq (2N+1)(K+1)$ and that there are at least $(2N+1)(K+1)$ linearly independent vectors in the set $\{\beta(s_j)\}_{j=1}^{J}$. This condition can be somewhat restrictive, but \eqref{eq:L2condbound} is only an upper bound on $\kappa(\bold{L}_2(\boldsymbol{\beta}))$, which may not be tight. In practice, $\kappa(\bold{L}_2(\boldsymbol{\beta}))$ takes on small ($<10$) values even if $\Gamma \succ 0$ is not satisfied. Furthermore, $\kappa(\Gamma)$ typically decreases with increasing spatial resolution $J$ of the projections.

\subsection{Condition numbers and view angle sampling schemes}
\label{subsec:ConditionNumbers}
Evaluation of the condition numbers $\kappa(L_1)$ and $\kappa(\bold{L}_2(\boldsymbol{\beta}))$ in the separate problems \eqref{eq:L1beta} and \eqref{eq:Linz} enables explicit analysis of the effects of the view angle sampling scheme on stability. 
We consider three schemes: (i) progressive with $\Delta_\theta = \frac{2\pi}{P}$; 
(ii) random, $\theta_p \sim \mathcal{U}[0,2\pi), \forall p$; and (iii) bit-reversed in $[0,2\pi)$,
with angles obtained by the reversal of the binary representations of the progressive scheme \cite{chan2012influence}.
The intervals are selected as $[0, \pi)$ when $\pi$-symmetry is exploited.
We study the dependence of the condition numbers
$\kappa(L_1(Z))$, and $\kappa(\bold{L}_2(\boldsymbol{\beta}))$ on  view-angle sampling scheme, for  fixed  $K,\,\, N,\,\, P$. For the study of $L_1$, the temporal functions in $\Psi$ were set using orthonormal Legendre polynomials of increasing order. (Because their version sampled uniformly at $P$ points is not exactly orthonormal, it was orthonormalized by Gram-Schmidt.)
$\bold{L}_2(\boldsymbol{\beta})$,  the interpolator matrix is set to  $U \in \Real^{2P \times d}$ with elements $U_{ij} \sim \mathcal{N}(0,1)$ independent and identically distributed (iid) and 
$\beta_{i} \sim \mathcal{N}(0,1)$, iid.

\begin{table}[hbt]
\setlength{\tabcolsep}{4pt}
\renewcommand{\arraystretch}{0.85}
\centering
\begin{tabular}{@{}lccc@{}}
\toprule
\multicolumn{1}{c}{ }&\multicolumn{1}{c}{Progressive}& \multicolumn{1}{c}{Random}& \multicolumn{1}{c}{Bit-reversed}\\
\cmidrule(r){1-1}\cmidrule(lr){2-2}\cmidrule(lr){3-3}\cmidrule(lr){4-4}
  $\kappa(L_1)$ no symm./symm. & 4.2e+16 / 1.8e+16 & 103.2 / 8.3 & \textbf{11.7} / \textbf{3.0} \\
  \cmidrule(r){1-1}
  $\kappa(\bold{L}_2)$ & 1.2 & 1.2 & 1.2  \\
 \bottomrule
\end{tabular}
\caption{\small Condition numbers  for different angular sampling schemes for $K$=5, $N$=28, $P$=512. For the random scheme, $\kappa(L_1)$ and $\kappa(\bold{L}_2)$ are the best out of 1000 different realizations.}
\label{tab:L1_L2_cond_nums_view_sch}
\end{table}

Table \ref{tab:L1_L2_cond_nums_view_sch} shows that for the given set of parameters, although $\kappa(\bold{L}_2)$ is not affected by the sampling scheme,  $\kappa(L_1)$ depends on it significantly, and a naive progressive scheme results in a practically singular model, whereas the bit-reversed scheme improves substantially over the random case. 
\section{Recovery Algorithms}
In view of unavoidable noise and model approximations, we solve \eqref{eq:L1beta} in the least-squares sense. Define the loss function and the least-squares problem as
\begin{align} \label{eq:LS-beta}
    \mathcal{L}(\beta(s), Z) &=  
    || \hat{g}(s) - L_1(Z)\beta(s) ||_2^2\\
\label{eq:LS-bilinear}
   \min_{\substack
    {Z \in \mathbb{R}^{d \times (K+1)} ,\,\, \\ \beta(\cdot) : \mathbb{R} \rightarrow \mathbb{R}^{(2N+1)(K+1)}
    }
    } 
    &\sum_s  
    \mathcal{L}(\beta(s), Z).
\end{align}

Problem \eqref{eq:LS-bilinear} can be solved using the Variable Projection method (VarPRO) \cite{golub1973differentiation}. 
First, because the summed term $\mathcal{L}(\beta(s), Z)$ is nonnegative, for fixed $Z$, \eqref{eq:LS-bilinear}
is minimized by minimizing 
 \eqref{eq:LS-beta}
 pointwise w.r.t $\beta(s)$, so that
 \eqref{eq:LS-bilinear}
becomes
\begin{equation}
    \min_{Z \in \mathbb{R}^{d \times (K+1)}}
    \sum_s \Bigg[ \min_{\beta(s) \in \mathbb{R}^{(2N+1)(K+1)}} \mathcal{L}(\beta(s),Z) \Bigg]. \label{eq:nestedmin}
\end{equation}

Now, the inner minimization in \eqref{eq:nestedmin}
is a least squares problem by $\beta(s)$, optimized for ${\beta}^*(s) = L_1^{\dagger}(Z)\hat{g}(s)$.
Inserting this into the objective function and simplifying, we obtain
\begin{equation}
\begin{aligned}
    \mathcal{L}({\beta}^*(s),Z) &= \mathrm{tr}\{P_{R^\perp(L_1(Z))}\hat{g}(s)\hat{g}(s)^T\}, 
\end{aligned}
\end{equation}
where $P_{R^\perp(L_1(Z))}$ is the orthogonal projection matrix onto the orthocomplement of the range space $R \left( L_1(Z) \right)$. 
of $L_1(Z)$.
Then \eqref{eq:nestedmin} reduces to
    \begin{align}
        \min_{Z \in \mathbb{R}^{d\times(K+1)}} \mathrm{tr}\{P_{R^\perp(L_1(Z))} \Xi(\hat{\boldsymbol{g}}) \}, \label{eq:minZ}
    \end{align}
where $P_{R^\perp(L_1(Z))}$ is the orthogonal projection matrix onto the orthocomplement of the range space $R \left( L_1(Z) \right)$ and where $\Xi(\hat{\boldsymbol{g}}) \triangleq \sum_s \hat{g}(s)\hat{g}(s)^T$ and $\Xi(\hat{\boldsymbol{g}}) \in \mathbb{R}^{P \times P}$  without $\pi$-symmetry, or $\Xi(\hat{\boldsymbol{g}}) \in \mathbb{R}^{2P \times 2P}$ with it.
In this form, \eqref{eq:minZ} reveals an important fact: the optimum $Z$, which determines the temporal functions $\psi_k(t)$, does not depend on the detailed measurements $\hat{g}(s)$. Instead, it only depends on the data matrix $\Xi(\hat{\boldsymbol{g}})$, which aggregates all the measurement information. 

As written, \eqref{eq:minZ} does not have a unique solution w.r.t. $Z$, because $R \left( L_1(Z) \right)$
is invariant to scaling of the columns of $Z$. 
We remove this ambiguity by constraining $Z$ 
to have orthonormal columns, 
$Z^TZ = I$.
In the implementation we use a penalized form and gradient descent for minimization.

\section{Experiments} \label{sec:experiments}
We compare ProSep with the recent GMLR method \cite{Hyder2020}, an extension of deep image prior \cite{ulyanov2018deep} to video.
GMLR uses a generative model
to map latent codes $\zeta_t$ to images $f_t$ to reconstruct a video from incomplete measurements. It performs joint optimization of latent codes and the generator parameters to match the measured data, while enforcing smoothness on the sequence of $\zeta_t$ to achieve temporal smoothness in $f_t$. 
In our application of  GMLR   to dynamic tomography it solves 
\begin{align*}
& \hspace{1cm} \min _{\zeta_{1}, \ldots, \zeta_{T} ; \gamma}  \mathcal{L}(\boldsymbol{\zeta}, {\gamma}) \quad
\text{s.t.}\quad \mathrm{rank}(\boldsymbol{\zeta}) = r \\
& \text{where} \quad \mathcal{L}(\boldsymbol{\zeta}, \gamma) =\\ &\lambda \sum_{t=1}^{P}\left\|g_{t}-\mathcal{R}_{\theta(t)} G_{\gamma}\left(\zeta_{t}\right)\right\|_{2}^{2}+ 
(1-\lambda) \sum_{t=1}^{P-1} \left\|\zeta_{t+1}-\zeta_{t}\right\|_{2}^{2}
\end{align*}
where $g_t = \mathcal{R}_{\theta(t)}f_t$, $G_\gamma$ is the generator with parameters $\gamma$, $\boldsymbol{\zeta} = [\zeta_1 \,\, \zeta_2 \,\, \ldots \zeta_P]$, and $\lambda$  controls the similarity of consecutive latent codes. In \cite{Hyder2020}  the architecture of $G_\gamma$ is the DCGAN \cite{radford2015unsupervised} for a $64 \times 64$ image with minor changes. Since our experiments 
consider $f \in \mathbb{R}^{128\times 128}$, another upsampling layer was added to the original $G_\gamma$ configuration \cite{HyderGithub2019}.

We present the reconstruction results for ProSep
with and without leveraging the $\pi$-symmetry. Both settings use the bit-reversed angular scheme between $[0,\pi]$ and $[0,2\pi]$, respectively. The interpolator $U$ is set to a cubic spline interpolator for both cases. The Adam \cite{kingma2014adam} algorithm was used for optimizing $Z$ with learning rate of $0.2$. 
For GMLR, the learning rates for $\gamma$ and $\zeta$ were kept as in the posted code \cite{HyderGithub2019}. 
For each $P$,
the model was trained for
$6(10^{4})$ steps with $\lambda = 0.5$ and rank $r=4$. 
In all experiments, we use the synthetic dynamic object shown in Fig.~\ref{fig:phantom_evolution}. It is based on a 128x128 CT slice of a walnut \cite{lahiri2022sparse}, to which we applied a time-varying locally affine warp \cite{pwise_affine}.

Table \ref{tab:avg_rec_acc} shows the reconstruction PSNR (in dB), SSIM and mean absolute error (MAE) values. Using $\pi$-symmetry for $P$ views between $[0,\pi)$ allows larger $N$ and $K$ for stable recovery and thus improves the accuracy significantly for each $P$. 
GMLR, ProSep, and ProSep-Symm, are compared in Fig.~\ref{fig:phantom_comparison}. Both error figures and quantitative metrics show competitive performance of ProSep relative to GMLR, and better performance of ProSep-Symm. Furthermore, the ProSep and ProSep-Symm reconstructions of the arc-shaped feature in the zoom-in images are sharper.

\begin{figure}[htbp!]
\centering
\setlength{\tabcolsep}{-0.02cm}
\renewcommand{\arraystretch}{0.1}
\begin{tabular}{ccc}
\includegraphics[width=.22\linewidth]{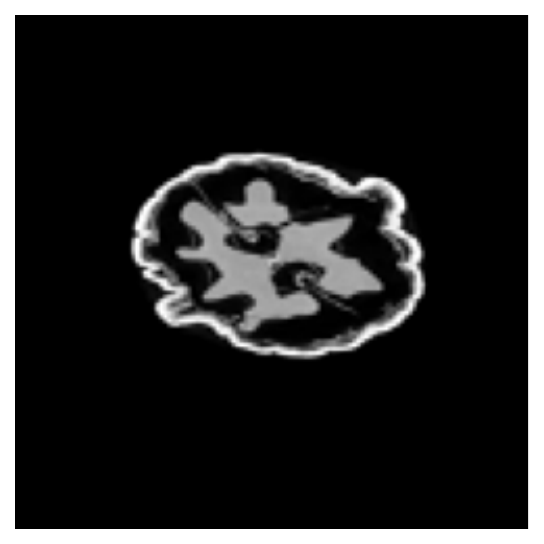} & 
\includegraphics[width=.22\linewidth]{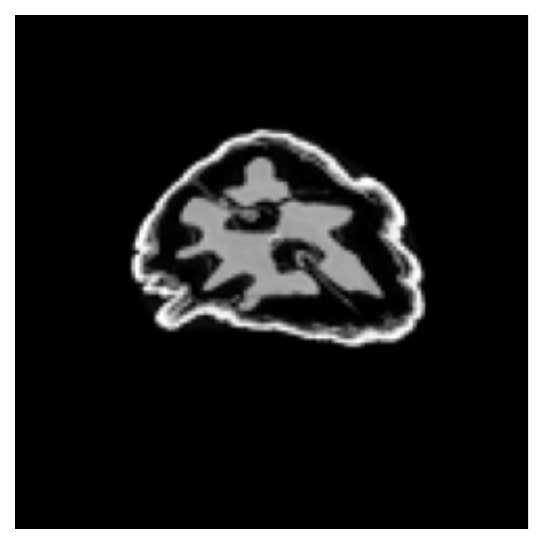} &
\includegraphics[width=.22\linewidth]{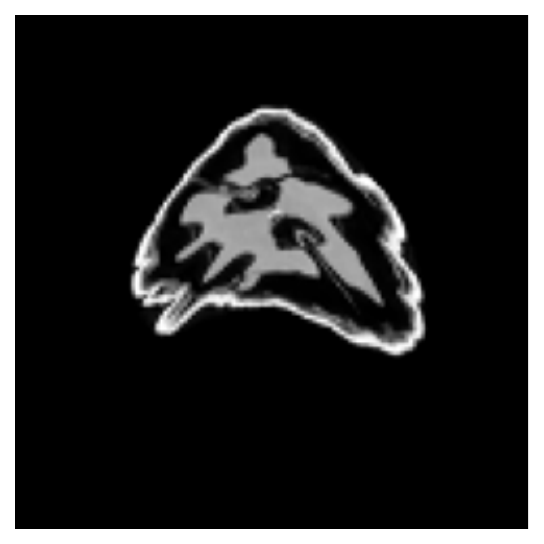}\\
(a)&(b)&(c)
\end{tabular}
\caption{\small The object undergoing pointwise affine transform at time instances (a) $t=0$, (b) $t=P/2$ and (c) $t=P$.}
\label{fig:phantom_evolution}
\end{figure}

\begin{table}[hbt!]
\setlength{\tabcolsep}{5pt}
\renewcommand{\arraystretch}{0.9}
\centering
\begin{tabular}{@{}clcccccc@{}}
\toprule
\multicolumn{1}{c}{$P$}&\multicolumn{1}{c}{Method}&\multicolumn{1}{c}{$K$}&\multicolumn{1}{c}{$N$}&\multicolumn{1}{c}{$d$}&\multicolumn{1}{c}{PSNR$\,\,$(dB)}& \multicolumn{1}{c}{SSIM}& \multicolumn{1}{c}{MAE}\\
\cmidrule(r){1-1}\cmidrule(lr){2-2}\cmidrule(lr){3-3}\cmidrule(lr){4-4}\cmidrule(lr){5-5}\cmidrule(lr){6-6}\cmidrule(lr){7-7}\cmidrule(lr){8-8}
  256 & GMLR  & - & - & - & 27.3 & 0.783 & 0.027 \\
  & ProSep & 3 & 24 & 4 & 26.9 & 0.894 & 0.022 \\
  & ProSep symm. & 5 & 30 & 6 & \textbf{30.4} & \textbf{0.928} & \textbf{0.015} \\
  \cmidrule(r){1-1}\cmidrule(lr){2-2}\cmidrule(lr){3-3}\cmidrule(lr){4-4}\cmidrule(lr){5-5}\cmidrule(lr){6-6}\cmidrule(lr){7-7}\cmidrule(lr){8-8}
  512 & GMLR & - & - & - & 31.1 & 0.876 & 0.017 \\
  & ProSep & 5 & 28 & 6 & 30.5 & 0.944 & 0.014 \\
  & ProSep symm & 7 & 48 & 8 & \textbf{35.1} & \textbf{0.959} & \textbf{0.010} \\
  \cmidrule(r){1-1}\cmidrule(lr){2-2}\cmidrule(lr){3-3}\cmidrule(lr){4-4}\cmidrule(lr){5-5}\cmidrule(lr){6-6}\cmidrule(lr){7-7}\cmidrule(lr){8-8}
  1024 & GMLR & - & - & - & 36.7 & 0.925 & 0.009 \\
  & ProSep & 7 & 48 & 8 & 36.8 & 0.979 & 0.007 \\
  & ProSep symm. & 9 & 56 & 10 & \textbf{39.5} & \textbf{0.980} & \textbf{0.006} \\
 \bottomrule
\end{tabular}
\caption{\small Average reconstruction accuracies for the complete 3D time-varying phantom in Fig. \ref{fig:phantom_evolution} for different $P$. ``ProSep symm" employs the $\pi$-opposite projection symmetry.}
\label{tab:avg_rec_acc}
\end{table}

\begin{figure}[htbp!]
\centering
\setlength{\tabcolsep}{-0.01cm}
\renewcommand{\arraystretch}{0.1}
\vspace{0cm}
\begin{tabular}{cccccc}
& 
{\small \textbf{GMLR}} & {\small\textbf{ProSep}} & \hspace*{-1.0ex}{\small \textbf{ ProSep-Symm}} & {\small \textbf{True Recon}} 
& 
\\
\rotatebox{90}{\makecell{
$\quad\,\,$Recon
}} &\includegraphics[width=.22\linewidth, trim={-0.1cm -0.1cm -0.1cm -0.1cm},clip]{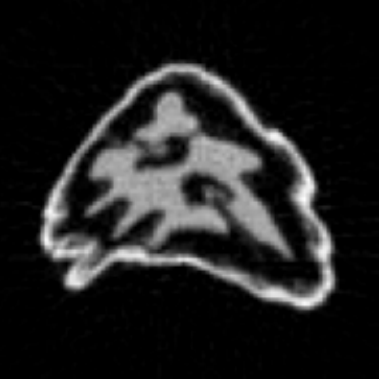} & 
\includegraphics[width=.22\linewidth, trim={-0.1cm -0.1cm -0.1cm -0.1cm},clip]{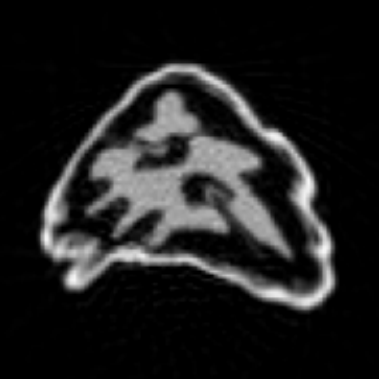} &
\includegraphics[width=.22\linewidth, trim={-0.1cm -0.1cm -0.1cm -0.1cm},clip]{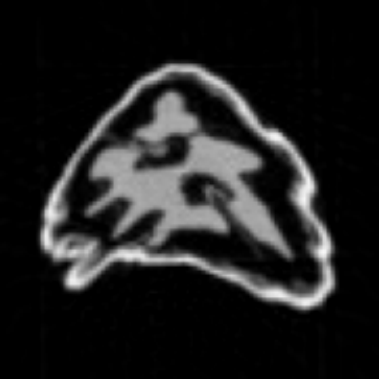} &
\includegraphics[width=.22\linewidth, trim={-0.1cm -0.1cm -0.1cm -0.1cm},clip]{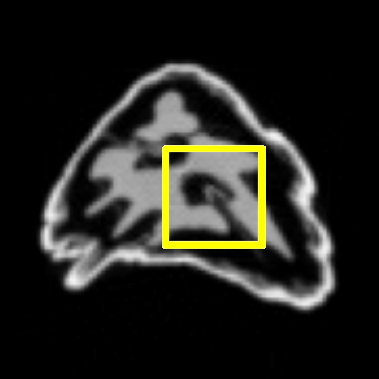} &
\hspace{-0.05cm}\includegraphics[width=.074\linewidth, trim={-0.1cm -0.1cm -0.1cm -0.1cm},clip]{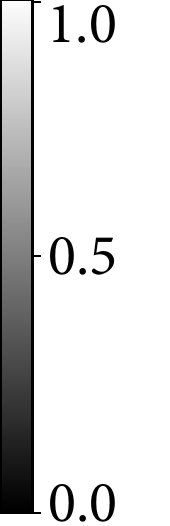} 
\vspace{-0.0cm}
\\

\rotatebox{90}{\makecell[c]{\centering
$\quad$Abs Error
}} &\includegraphics[width=.22\linewidth, trim={-0.1cm -0.1cm -0.1cm -0.1cm},clip]{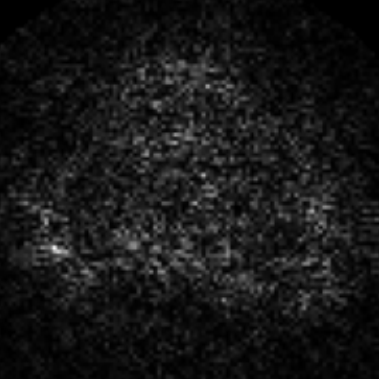} & 
\includegraphics[width=.22\linewidth, trim={-0.1cm -0.1cm -0.1cm -0.1cm},clip]{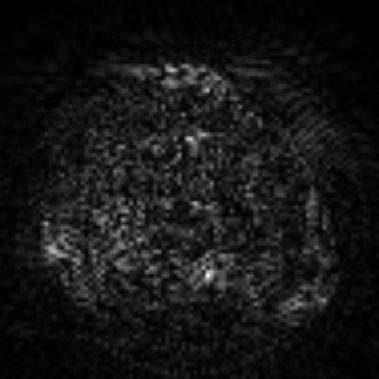} &
\includegraphics[width=.22\linewidth, trim={-0.1cm -0.1cm -0.1cm -0.1cm},clip]{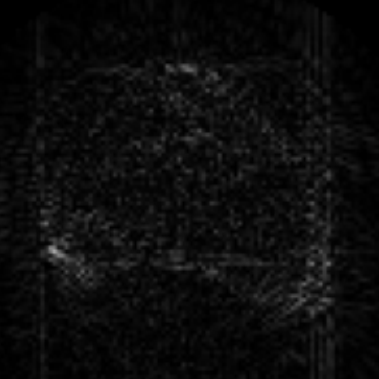} &
\includegraphics[width=.22\linewidth, trim={0.14cm 0.14cm 0.14cm 0.14cm},clip]{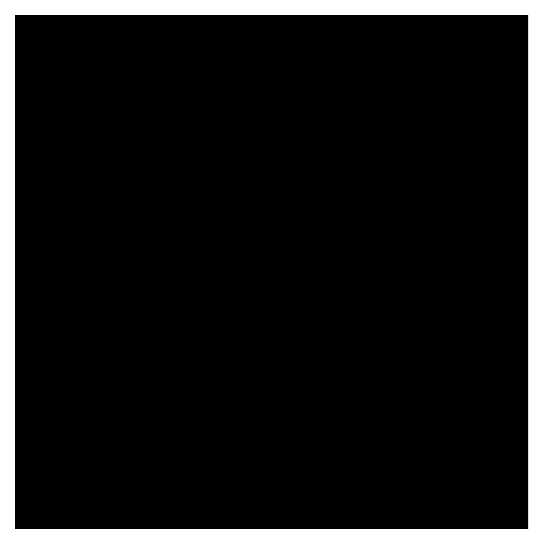} &
\hspace{0.01cm}\includegraphics[width=.08\linewidth, trim={-0.11cm -0.1cm -0.1cm -0.1cm},clip]{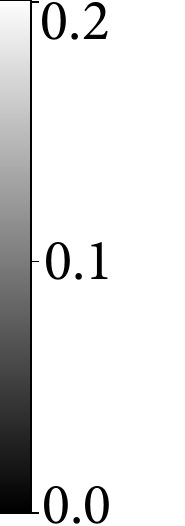} 
\vspace{-0.0cm}
\\
\rotatebox{90}{\makecell{
$\quad$Zoom-in
}} &\includegraphics[width=.22\linewidth, trim={-0.04cm -0.04cm -0.04cm -0.04cm},clip]{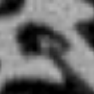} & 
\includegraphics[width=.22\linewidth, trim={-0.04cm -0.04cm -0.04cm -0.04cm},clip]{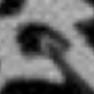} &
\includegraphics[width=.22\linewidth, trim={-0.04cm -0.04cm -0.04cm -0.04cm},clip]{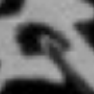} &
\includegraphics[width=.22\linewidth, trim={-0.04cm -0.04cm -0.04cm -0.04cm},clip]{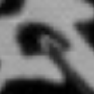}&
\hspace{-0.18cm}\includegraphics[width=.074\linewidth, trim={-0.1cm -0.1cm -0.1cm -0.1cm},clip]{Figures/cbars/recon_cbar_modified_normalized.pdf} 

\end{tabular}
\caption{\small Row 1: reconstructions of the object in Fig. \ref{fig:phantom_evolution} at $t=990/1024$ for different methods;  Row 2:  absolute deviations from ``True Recon"; Row 3: zoom-in of the yellow box region. 
}
\label{fig:phantom_comparison} 
\end{figure}

\section{Conclusions}
We introduced a special partially separable model for the harmonic expansion of the projections of a dynamic object, and formulated the recovery from time-sequential projections as a bilinear inverse problem. We analyzed uniqueness and stability of recovery using the proposed method and compared it with GMLR, a recent generative model for video reconstruction. Unlike GMLR, in its current form the proposed method does not use a spatial prior, yet in our experiments it was competitive with or better than GMLR.

Ideas for future study include  an object-adaptive view angle sampling scheme, and a generative model to introduce a spatial prior for the time-varying object to improve the accuracy of the reconstructions with reduced sampling requirements.

\bibliographystyle{unsrt}
\bibliography{TVT_bib}

\clearpage
\section{Appendix}
In this section, we establish some of the proofs of the theorems and derivations that are used in the paper.

\subsection{Partially-Separable Model for Affine Motion}
\label{app:par-sep_model_for_aff_motion}
Throughout this section we
consider a nominal object $f_0(\xb) \in \mathcal{L}_2$, with $\xb = (x_1, x_2) \in \Real^2$, vanishing outside a disk of radius $L$, and essentially bandlimited to spatial bandwith $B$ radians, that is
\begin{align}
    &f_0(\xb) = 0 \quad\text{for} \quad \|\xb\|_2 > L \\
    &\int_{\|\omegab\| > B} |F_0(\omegab)|^2 d \omegab \leq \epsilon_f 
\end{align}
where $F_0$ is the 2D Fourier transform of $f_0$, and $\epsilon_f \ll \|f_0\|_2^2$. We consider motions of the nominal object resulting in a time-varying object $f(\xb,t)$, which we assume remains bounded in the disk of radius $L$, that is, $f(\xb,t) = 0 \, \forall t$ for  $\|\xb\|_2 > L$.
\subsubsection{Time-Varying Translation}
\label{app:translation_par_sep}
Consider the nominal object translating in time along trajectory $\cb(t) = (c_1(t), c_2(t))$ of bounded extent, $\|\cb(t)\|_2 \leq c_{\max} \, \forall t $,
resulting in the time-varying object
\begin{align}
    f(\xb,t) &= f_0(\xb-\cb(t)) \\
    F(\omegab, t) &= F_0(\omegab)e^{-j(\omegab^T \cb(t))}. \label{eq:comp_exp_four_domain}
\end{align}
Using the assumptions and the Cauchy-Schwartz inequality, it follows that
$|\omegab^T \cb(t)| \leq Bc_{\max}$, suggesting the use of a $K$-th order Taylor series expansion of the complex exponential in \eqref{eq:comp_exp_four_domain}. To bound the remainder term, we use the integral form of the remainder in Taylor's theorem,
\begin{equation}
    e^{jx}=\sum_{k=0}^{K} \frac{(j x)^{k}}{k !}+\frac{j^{K+1}}{K !} \int_{0}^{x}(x-s)^{K} e^{j s} d s
\end{equation}
which yields
\begin{align}
\left| e^{jx} - \sum_{k=0}^{K} \frac{(j x)^{k}}{k !} \right| &\leq \frac{|x|^{K+1}}{(K+1)!} \nonumber \\ &\approx \sqrt{2 \pi (K+1)} \left(\frac{e|x|}{K+1}\right)^{K+1} , \label{eq:TaylorCompExp}
\end{align}
where the approximate form follows by Striling's approximation of the factorial.
It follows that the remainder is exponentially decaying for $K> e|x| -1$.
Applying \eqref{eq:TaylorCompExp} to the complex exponential in \eqref{eq:comp_exp_four_domain} we obtain that the remainder in its $K+1$ term expansion
is bounded by
\begin{align} \label{eq:remainTaylorCompExp}
    \left| e^{-j\omegab^T\cb(t)} - \sum_{k= 0}^{K} \frac{(-j\omegab^T\cb(t))^{k}}{k !} \right| &
    \leq \frac{|Bc_{\max}|^{K+1}}{(K+1)!} \\
    & \hspace{-1em} \approx 
    \sqrt{2 \pi (K+1)} \left(\frac{Bc_{\max}}{K+1}\right)^{K+1} .
    \nonumber
\end{align}
This remainder is exponentially decaying for $K> Bc_{\max}-~1$. Because it holds pointwise for each $\omegab$, it implies that the corresponding $K$-th order expansion for $F(\omegab,t)$ has relative rms error bounded by the same remainder bound.

The resulting expansion leads to the following approximation in the spatial domain
\begin{align} \label{eq:ParSepTrans-space}
   {f}(\xb,t) & \approx \sum_{\substack{(\alpha_1 + \alpha_2)\leq K \\ \alpha_1, \alpha_2 \geq 0}} \hspace{-1em} (-1)^{(\alpha_1 + \alpha_2)} c(t)^\alpha  \left({\frac{\Dc^\alpha}{\alpha_1!\alpha_2!} }
    f_0(\xb) \right)
\end{align}
where we use the index notation
\begin{align}
    \cb^\alpha &\triangleq c_1^{\alpha_1} c_2^{\alpha_2} \nonumber \\
    \Dc^\alpha &\triangleq \frac{\partial^{(\alpha_1+\alpha_2)}}{\partial x_1^{\alpha_1} \partial x_2^{\alpha_2}}
    \label{eq:multi_idx_notation}
\end{align}
By Parseval's identity, the relative rms truncation error of the expansion in \eqref{eq:ParSepTrans-space} is bounded by the right-hand side of \eqref{eq:remainTaylorCompExp}.
Hence, $f$ admits a partially separable expansion with $(K+1)(K+2)/2$ terms, with possibly moderate $K$. 

The number of significant terms in the expansion depends on the spatial bandwidth $B$ of $f_0$ and the magnitude $c_{\max}$ of the motion.

\label{sec:nominal_f_traj}
\subsubsection{Time-Varying Scaling}
{Consider  the nominal object $f_0(\xb)$ undergoing scaling by diagonal matrix $\Cmat(t) = \diag \left(c_1(t), c_2(t)\right)$
to $f(\xb,t) = f_0(\Cmat(t)\xb)$.}
This case too can be represented using the partially-separable model. To see this, we use the 2D Mellin transform \cite{brychkov1992multidimensional}
\begin{equation}
    \hat{f}(s,\tau) \triangleq \int_{0}^{\infty}\int_{0}^{\infty} x_1^{s-1}x_2^{\tau-1}f(\xb)dx_1dx_2
\end{equation}
which is well-defined and has an inverse if the integral
\begin{align}
\int_{0}^{\infty}\int_{0}^{\infty} |f(\xb)|x_1^{s-1}x_2^{\tau-1}dx_1 dx_2
\end{align}
converges for all $s$ s.t $\alpha_1 < \Re(s) < \beta_1$ and all $\tau$ s.t $\alpha_2 < \Re(\tau) < \beta_2$ for some $\alpha_1  < \beta_1$ and $\alpha_2  < \beta_2$.

Since $f(\xb,t)$ is defined on a centered disk, we cannot apply the Mellin transform, which is only defined for non-negative arguments, to it directly. Instead, we consider the part of $f(\xb,t)$ supported in the first quadrant, and to simplify the notation we still denote it by $f(\xb,t)$ and $f_0(\xb)$.
The extension to include the other parts is discussed later. 
For our setting, {since $f$ is supported on a disk of radius $L$ and $f \in \Lc_2$, this condition holds for $0.5 < \Re(s) < \infty$ and $0.5 < \Re(\tau) < \infty$ .}

The Mellin transform $\hat{f}$ of the object is given by
\begin{equation}
    \hat{f}(s,\tau, t) = c_{1}(t)^{-s}c_{2}(t)^{-\tau}\hat{f}_0(s,\tau)
\end{equation}

Again, using a {truncated series expansion for the exponential term, we obtain
\begin{align*}
    & \hat{f}(s,\tau,t) \approx  \sum_{k=0}^K \frac{(-1)^k}{k!}\left(s \ln c_{1}(t) +\tau \ln c_{2}(t) \right)^k \hat{f}_0(s,\tau) \\ \quad
    &=   \sum_{\substack{(\alpha_1 + \alpha_2) \leq k \\ \alpha_1, \alpha_2 \geq 0}} \hspace{-1em}  \frac{(-1)^{\alpha_1+\alpha_2} }{\alpha_1 ! \alpha_2!} \left(s \ln c_{1}(t)  \right)^{\alpha_1} \left(\tau \ln c_{2}(t) \right)^{\alpha_2} \hat{f}_0(s,\tau)
\end{align*}
}
Using the differentiation property for the 2D Mellin transform 
\begin{equation}
    \left( \xb_1\frac{d}{d\xb_1} \right)^{\alpha_1} \left( \xb_2\frac{d}{d\xb_2} \right)^{\alpha_2} f(\xb) \leftrightarrow (-1)^{\alpha_1+\alpha_2} s^{\alpha_1} \tau^{\alpha_2} \hat{f}(s,\tau),
\end{equation}
yields 
\begin{align}
    f(\xb,t) & \approx \sum_{\substack{(\alpha_1 + \alpha_2) \leq k \\ \alpha_1, \alpha_2 \geq 0}} \ln^\alpha \Cmat(t) \left( \frac{1}{\alpha_1!\alpha_2!}(\xb \Dc)^\alpha 
    f_0(\xb) \right)
    \label{eq:kth_order_spatial_approx_scaling}
\end{align}
using the notation in \eqref{eq:multi_idx_notation} with 
\begin{equation}
    \ln^\alpha\Cmat = \ln^{\alpha_1}c_1(t) \ln^{\alpha_2}c_2(t).
\end{equation}
In addition to the spatial bandwidth of $f_0$ and the log-magnitude of the motion components $c_1(t)$ and $c_2(t)$, the number of significant terms in the expansion depends also on the radius $L$ of the support of the object $f_0$. Time-varying scaling factors $c_1(t)$ and $c_2(t)$ being closer to 1, lower spatial bandwidth, and a smaller radius $L$ would result in fewer significant terms in the expansion.

 To incorporate the parts of the object in the remaining quadrants to the analysis, we  use the following steps. The object $f_0$ is decomposed into the sum of 4 pieces as $f_0(\xb) = \sum_{m,n=0}^1 f_{mn}(\xb)$ where $f_{mn}$ is supported in the $mn$-th quadrant in binary notation. Then, we  represent the versions of these decomposed parts of the object after reflection  into the first quadrant as $\tilde{f}_{mn}(\xb) =~f_{mn}((-1)^m x_1, (-1)^n x_2)$ and apply \eqref{eq:kth_order_spatial_approx_scaling} on each $\tilde{f}_{mn}$. The object is finally  reconstituted from its subparts as $f(\xb,t) =~\sum_{m,n=0}^1 \tilde{f}_{mn}((-1)^m x_1, (-1)^n x_2, t)$ by incorporating the necessary sign changes in the term $(\xb \mathcal{D})^\alpha$ as $\left( (-1)^m \xb_1\frac{d}{d\xb_1} \right)^{\alpha_1} \left( (-1)^n \xb_2\frac{d}{d\xb_2} \right)^{\alpha_2}$.

Although this partition creates discontinuities at $x_1 = 0$ and $x_2 = 0$, since the derivatives are multiplied with $x_1$ and $x_2$ at these points, these discontinuities do not constitute a problem. Furthermore, the composition of the object from the four quadrants does not increase the number of terms in the expansion \eqref{eq:kth_order_spatial_approx_scaling}, because the expansions for all four quadrants share the same temporal functions $\ln^\alpha \Cmat(t)$.
\smallskip

\subsubsection{Time-Varying Rotation}
Consider the nominal object $f_0(x)$ rotated by angle $\theta(t)$. Using the polar representation of $f(\xb,t)$
\begin{equation}
    \begin{aligned}
        f_{\mathrm{pol}}(r,\phi,t) = f_{0,\mathrm{pol}}(r,\phi - \theta(t))
    \end{aligned}
\end{equation}
 the rotation angle $\theta(t)$ acts as a translation in the angular coordinate. {We therefore adopt the same approach used for a time-varying translation in Appendix~\ref{app:translation_par_sep} 1).}
Computing the Fourier transform of  $f_{pol}$ which is $2\pi$-periodic in $\phi$,  w.r.t. $\phi$, yields
\begin{align}
    F_{pol}(r,\omega_\phi,t) = F_{0,pol}(r,\omega_\phi,t)e^{-j\omega_\phi \theta(t)}.
\end{align}
Expanding the complex exponential using a $k$-th order Taylor series as in Appendix \ref{app:translation_par_sep}, we obtain a similar bound on the remainder term {
\begin{align}
    \left| e^{j\omega_\phi \theta(t)} - 
    \sum_{k= 0}^K \frac{1}{k!}
    (j\omega_\phi\theta(t))^k \right| 
    &\leq \frac{\left| \omega_\phi \theta(t) \right|^{K+1}}{ (K+1)! } \\
    &\leq \frac{|B_\phi \theta_{\max}|^{K+1}}{(K+1)!}
\end{align}
where $|\theta(t)| \leq \theta_{\max} \,\, \forall t$ and $B_\phi = BL$} is the angular bandlimit of the object $f_{0,pol}$.

Finally, the resulting expansion in the spatial domain is given by
{
\begin{equation}
   {f}_{pol}(r,\phi,t) \approx \sum_{k=0}^K (-1)^k \theta^k(t)
   \left( \frac{1}{k!} \frac{d^k}{d\phi^k}  f_{0,pol}(r,\phi) \right).
\end{equation}

It follows that in the  case of rotation, $f_{pol}$ admits a partially separable representation with $K+1$ significant terms,  whose number depends on the bandwidth of the object, its radius $L$, and the magnitude of the rotation. Again, the bound on the truncation error decays exponentially with decreasing values of these parameters. }

\subsection{Proof of Theorem \ref{thm:proj_special_par_sep}}

Before proving the result we state the following lemma.

\begin{lemma}
\label{lem:proj_norm_bnd_lemma}(\cite{epstein2003introduction}, Proposition 6.6.1.)
Suppose that object $f \in \mathcal{L}_2(\mathbb{R}^2)$  vanishes outside a disk of radius $L$. Then, for each $\theta$, we have the estimate
\begin{equation*}
\label{eq:radon_err_bnd}
    \int_{-\infty}^{\infty} |(\mathcal{R}f)(s,\theta)|^2 ds \leq 2L||f||_2^2.
\end{equation*}
\end{lemma}

Denote the first term in \eqref{eq:f_parsep_w_err} by $\tilde{f}(\xb,t)$. Then its Radon transform  is given by
\begin{equation}
\label{eq:radon_obj_term}
    (\mathcal{R}\tilde{f})(s, \theta,t) =
    \sum_{k=0}^{K-1} q_k(s,\theta) \psi_k(t)
\end{equation}
where $q_k \triangleq \mathcal{R}f_k$. Next, expanding $q_k$ in a Fourier series yields 
\begin{equation}
\label{eq:q_fourier_series}
    q_k(s,\theta) = \sum_n \beta_{n,k}(s) e^{jn\theta}.
\end{equation}
Substituting \eqref{eq:q_fourier_series} into \eqref{eq:radon_obj_term} 
and switching the order of summation yields 
    \begin{align}
        (\mathcal{R}\tilde{f})(s, \theta,t) =\sum_n h_{n}(s,t) e^{jn\theta} \label{eq:special_par_sep_proj}
    \end{align}
    where $h_{n}(s,t)$ is given by \eqref{eq:par-sep-special}. 

Next, we define the Radon transform of the second term in \eqref{eq:f_parsep_w_err} as $\gamma_g(s,\theta,t) = (\mathcal{R}\gamma_f )(s,\theta,t)$. Then, applying Lemma \ref{lem:proj_norm_bnd_lemma} to $\gamma_f(.,t)$ for fixed $t$
yields
\begin{equation*}
    \int_{-\infty}^{\infty} |\gamma_g(s,\theta,t)|^2 ds \leq 2L \int_{-\infty}^{\infty} |\gamma_f(\xb,t)|^2 d \xb
\end{equation*}
 and integrating both sides over $\theta$ and $t$
\begin{equation}
\begin{aligned}
\label{eq:err_bnd_proj}
    ||\gamma_g||_2^2 \leq 2 \pi L ||\gamma_f||_2^2 
    \leq 2 \pi L \epsilon_f.
\end{aligned}
\end{equation}
    
Combining \eqref{eq:special_par_sep_proj} and \eqref{eq:err_bnd_proj} shows that the projections admit the representation \eqref{eq:proj_par_sep_rep_w_err} with an error term bounded as $||\gamma_g||_2^2 \leq 2\pi L\epsilon_f$ and $h_n(s,t)$ represented by the \emph{special} partially separable model \eqref{eq:par-sep-special}. \qedsymbol

\subsection{Derivation of $\bold{L}_2(\boldsymbol{\beta})$}
\label{app:C}
To formulate the model linear in $Z$ of Sec.~\ref{sec:bilinear_prob_and_nec_cond} we manipulate the expression for the $i$-th element of $\hat{g}(s)$, 
\begin{equation}
    \begin{aligned}
    \label{eq:g_s_i_th_element}
        \hat{g}(s)_i = \hat{g}(s)_i^T &= {\beta}(s)^T L_{1_{i:}}^T(Z) \\
        &= {\beta}(s)^T (A_i^T \otimes \widehat{U}_{i:})\bold{z} \\
        &= {\beta}(s)^T A_i^T ( I_{K+1} \otimes \widehat{U}_{i:} ) \zb.
    \end{aligned}
\end{equation}
where $A_i$ is defined in \eqref{eq:A_i_def}.

Define  $\hat{g}^{(i)} \in \mathbb{R}^J$ by 
$\hat{g}_s^{(i)} \triangleq \hat{g}(s)_i$, that is, $\hat{g}^{(i)}$ is a vector containing the $i$-th elements of all $\hat{g}(s)$. This vector can be obtained as
\begin{equation}
\label{eq:L2form}
    \hat{g}^{(i)} = L_2({\beta})^{(i)} \zb \quad i=1,\ldots,P
\end{equation}
where matrix $\boldsymbol{\beta} \in \mathbb{R}^{(2N+1)(K+1) \times J}$ contains ${\beta}(s_j), \,\, j\in\{1, \ldots,J\}$ as columns and $L_2({\beta})^{(i)} \in \mathbb{R}^{J \times d(K+1)}$ is the stacking in $s$ of row vectors 
\begin{equation}
    L_2({\boldsymbol{\beta}(s)})^{(i)} \triangleq {\beta}(s)^T A_i^T ( I_{K+1} \otimes \widehat{U}_{i:} ).
\end{equation}

Then, stacking the vectors $\hat{g}^{(i)}$ and operators $L_2({\beta})^{(i)}$ for $i \in \{1, \ldots, 2P\}$, we obtain the problem \eqref{eq:Linz} linear in $Z$
where $\hat{g}$ and  $\bold{L}_2(\boldsymbol{\beta})$ are defined as
\begin{equation}
    \begin{aligned}
        \hat{\gb} = \begin{bmatrix}
        \hat{g}^{(1)} \\
        \vdots \\
        \hat{g}^{(2P)}
         \end{bmatrix} 
        \quad \bold{L}_2(\boldsymbol{\beta}) = \begin{bmatrix}
        L_2({{\beta}(s_1)})^{(1)} \\
        \vdots \\
        L_2({{\beta}(s_J)})^{(2P)}
        \end{bmatrix}.
        \label{eq:g_and_L2_stacked_def}
    \end{aligned}
\end{equation}

\subsection{Proof of Theorem \ref{thm:L1_full_rank}}

We first state and prove some preliminary results.  We use the property of Kruskal rank of  matrix $M$, or $\krank(M)$, defined as the maximal number $k$ such that any subset of $k$ columns of $M$ are linearly independent.

\begin{theorem}
\label{thm:L1_full_rank_V_krank}
Suppose $V \in \mathbb{C}^{2P\times n}$ and $V^T$ has Kruskal rank $\mathrm{krank}(V^T) = n$, and $\Psi \in \Real^{P \times L}$ has elements independently and identically distributed as $N(0,1)$. Let matrix $\widehat{\Psi} \triangleq [\Psi^T,  \Psi^T]^T$ have columns $\widehat{\psi_i}$.
Then, if $2P \geq nL$, the matrix $M=\begin{bmatrix} \mathrm{diag}(\widehat{\psi}_1)V & \ldots & \mathrm{diag}(\widehat{\psi}_{L})V \end{bmatrix}$ has  full column rank w.p.1.
\end{theorem}

\begin{proof}
    Throughout the proof we assume $2P=nL$, so that $M$ is a square $nL \times nL$ matrix, and note that because adding rows to a matrix does not decrease column rank, the obtained results hold for $2P\geq nL$. We also assume that $n$ is even, and consider the odd $n$ case at the end. 
    
    We denote the set of integers $[1,2,\ldots,N]$ by $[N]$. To show the desired result, we follow an approach similar to the proof of Theorem 5.2 in \cite{pfister2020composition}.
    
    Partition $[P]$ into  $L$ disjoint sets {$\{J_l \subset~[P] \}_{l=1}^{L}$
     of equal size $|J_l| = n/2$}, and divide the matrix $V$ into two parts $V^T = [V^{(1)T}, \,\, V^{(2)T}]^T$ where $V^{(1)}, V^{(2)} \in \mathbb{C}^{P \times n}$.
     Denote the matrices containing the rows of $V^{(1)}$ and $V^{(2)}$ indexed by elements of {$J_l$ by $V^{(1)}_{J_l,:}, V^{(2)}_{J_l,:} \in \Complex^{(n/2) \times n}$.
     Then the condition $\krank(V^T) = n$ yields 
    \begin{equation}
    \rank\left(\begin{bmatrix}V^{(1)}_{J_l,:} \\ V^{(2)}_{J_l,:}\end{bmatrix}\right) = \rank\left(\begin{bmatrix}V^{(1)}_{J_l,:} \\ V^{(2)}_{J_l,:}\end{bmatrix}^T\right) =  n.
    \end{equation}
}
        Now, $\mathrm{rank}(M) = nL \iff \det M \neq 0$. Since $\det M$ is a multivariate polynomial in the $PL= nL^2/2$ entries  of ${\Psi}$ with  coefficients dependent only on the entries of $V$, it is either identically zero or its zero set is an affine algebraic set and thus a nowhere dense set of measure zero in $\Complex^{nL^2/2}$. Thus, it suffices to show $\det M \neq 0$ for a single choice of $\Psi$ \cite{harikumar1998fir, harikumar1999perfect, jiang2001almost}, demonstrating that the polynomial $\det M$ does not vanish identically.

    Permuting the rows of $M$ produces the following matrix
    \begin{align}
    \label{eq:M_row_permuted}
        \begin{bmatrix}
            D^{(1)}_{J_1}V^{(1)}_{J_1,:} & \ldots & D^{(L)}_{J_1}V^{(1)}_{J_1,:} \\
            D^{(1)}_{J_1}V^{(2)}_{J_1,:} & \ldots & D^{(L)}_{J_1}V^{(2)}_{J_1,:} \\
            \vdots & \ddots & \vdots \\
            D^{(1)}_{J_{L}}V^{(1)}_{J_{L},:} & \ldots & D^{(L)}_{J_{L}}V^{(1)}_{J_{L},:} \\
            D^{(1)}_{J_{L}}V^{(2)}_{J_{L},:} & \ldots & D^{(L)}_{J_{L}}V^{(2)}_{J_{L},:}
        \end{bmatrix}
    \end{align} 
    where {$D^{(l)}_{J_k} \triangleq \diag({\Psi}_{J_k,l})$. Setting ${\Psi}_{J_l,l} = \mathbf{1}_{n/2}$ and thus $D^{(l)}_{J_l} = I_{n/2}$, and ${\Psi}_{J_k,l} = \mathbf{0}_{n/2}$ and thus $D^{(l)}_{J_k} = 0$ for $l\neq k$} yields
    \begin{align} \label{eq:block-diag-V}
        \begin{bmatrix}
            V^{(1)}_{J_1,:} & \ldots & 0\\
            V^{(2)}_{J_1,:} & \ldots & 0\\
            \vdots & \ddots & \vdots \\
            0 & \ldots & V^{(1)}_{J_{L},:}\\
            0 & \ldots & V^{(2)}_{J_{L},:}
        \end{bmatrix}
    \end{align}
    which is a block diagonal matrix with each $n \times n$ block along the diagonal being full rank by the assumption $\krank(V^T) = n$. Thus, $M$ is full rank for this choice of ${\Psi}$, and hence has full rank for almost all $\Psi$ (i.e., generically)  and w.p. 1 for the random ${\Psi}$.

    For the odd $n$ case, if $L$ is even, we choose $2P = nL$. Then, consider two different complementary partitions of $[P]$ into $L$ subsets, $[P] = \bigcup_{k=1}^{L} J_k^{(i)}, i=1,2$ such that for each $k$, $|J_k^{(1)}| + |J_k^{(2)}| = n$.
    We apply the $i=1$ partition to the top $P$ rows of matrix $M$ and to $V^{(1)}$ and the $i=2$ partition to the bottom $P$ rows of matrix $M$ and to $V^{(2)}$. Repeating the previous argument involving permutation of the rows of $M$, the resulting matrix analogous to \eqref{eq:M_row_permuted} will have $k-l$ block given by
    \begin{equation}
        \begin{bmatrix}
            D^{(l)}_{J_k^{(1)}}V^{(1)}_{J_k,:}  \\[2ex]
            D^{(l)}_{J_k^{(2)}}V^{(2)}_{J_k,:} 
            \end{bmatrix} = \mathrm{block.diag} \left(D^{(l)}_{J_k^{(1)}}, D^{(l)}_{J_k^{(2)}} \right)
            \begin{bmatrix}
            V^{(1)}_{J_k,:}  \\ V^{(2)}_{J_k,:}
             \end{bmatrix}
    \end{equation}
    where $D^{(l)}_{J_k^{(i)}} \triangleq \diag({\Psi}_{J_k^{(i)},l})$. Setting 
    $[{\Psi}^T_{J_k^{(1)},l}, {\Psi}^T_{J_k^{(2)},l} ]^T= \delta[k-l]\mathbf{1}_{n}$
    and thus $\mathrm{block.diag} (D^{(l)}_{J_k^{(1)}}, D^{(l)}_{J_k^{(2)}} ) = \delta[k-l]I_n$, yields again a full-rank block diagonal matrix for the permuted $M$ as in \eqref{eq:block-diag-V}, establishing the result for $n$ odd and $L$ even.
    
    When both $n$ and $L$ are odd, we choose $2P = nL+1$  (which is the smallest integer value satisfying $ 2P \geq nL$), 
    and the additional row is discarded before partitioning $[P]$ and $[P-1]$ into $L$ sets for $V^{(1)}$ and $V^{(2)}$ as in the even $nL$ case since the extra row does not affect the full rankness of  matrix $M$. 
\end{proof}

\begin{corollary}
\label{cor:gauss_to_stiefel}
Theorem~\ref{thm:L1_full_rank_V_krank} also holds if instead of being a random normal matrix, $\Psi$ is drawn at random from an absolutely continuous probability distribution on the Stiefel manifold $V_{(L)}\left(\mathbb{R}^{P}\right)$.
\end{corollary}

\begin{proof}
There exists a column permutation matrix $\Pi_M$ such that
\begin{equation}
     M \Pi_M = \begin{bmatrix} \mathrm{diag}(v_1)\widehat{\Psi} & \ldots & \mathrm{diag}({v}_{2N+1})\widehat{\Psi} \end{bmatrix} .
\end{equation}
Consider the SVD  $\Psi = U \Sigma \Xi^*$, and let { $S \in \Real^{L \times L}$ be the block-diagonal matrix with $\Xi\Sigma^{-1}$ repeated $n$ times on its diagonal, i.e.}
\begin{equation}
\label{eq:matrix_S}
    S \triangleq \mathrm{block.diag} (\Xi\Sigma^{-1}, \ldots, \Xi\Sigma^{-1}).
\end{equation}
Then, 
\begin{equation}
    M \Pi_M S = \begin{bmatrix} \mathrm{diag}(v_1)\begin{bmatrix}U \\ U\end{bmatrix} & \ldots & \mathrm{diag}({v}_{2N+1})\begin{bmatrix}U \\ U\end{bmatrix} \label{eq:gauss_to_stiefel_1} \end{bmatrix} .
\end{equation}
{Applying another column permutation matrix $\Pi_S$ yields
\begin{align}
   M \Pi_M S\Pi_S  &= \begin{bmatrix} \mathrm{diag}(\widehat{u}_1)V & \ldots & \mathrm{diag}(\widehat{u}_{K+1})V \end{bmatrix}, \label{eq:gauss_to_stiefel_2} 
\end{align}
where $\widehat{u_l} \triangleq [u_l^T, u_l^T]^T$.
Now, because $S$ is an invertible matrix, and column permutations do not change the rank, we have
\begin{equation}
\label{eq:rank_eqv_corollary}
    \rank ( M \Pi_M S\Pi_S ) = \rank (M \Pi_M S) = \rank (M).
\end{equation}
}

By Theorem \ref{thm:L1_full_rank_V_krank} $M$ has full column rank w.p. 1 when $\Psi$ has elements i.i.d. distributed as $N(0,1)$. Now, for the SVD of $\Psi = U\Sigma\Xi^*$, it is well-known that the left singular vectors $U$ are distributed uniformly on the Stiefel manifold, {because the distribution of an iid Gaussian matrix is invariant to rotations (on both left and right)}. Then, using the second identity in \eqref{eq:rank_eqv_corollary},  matrix $M\Pi_M S$ has also full rank w.p. 1 for $U$ drawn uniformly at random from the Stiefel manifold $V_{(L)}\left( \mathbb{R}^{P} \right)$. Thus, selecting $\Psi = U$, and using the first identity in \eqref{eq:rank_eqv_corollary}  establishes the  corollary for the uniform distribution on the Stiefel manifold. However, this implies the same for any distribution that is absolutely continuous with respect to the latter.
\end{proof}

Now, having stated the preliminary results of Theorem \ref{thm:L1_full_rank_V_krank} and Corollary \ref{cor:gauss_to_stiefel}, we can prove the Theorem \ref{thm:L1_full_rank}.

Recall that $L_1 = \widehat{\Theta} \bullet \widehat{\Psi}$, where $\widehat{\Theta} \in \Complex^{2P \times(2N+1)}$, and $\widehat{\Psi} \in \Real^{2P \times (K+1)}$. Hence $L_1 \in \Complex^{2P \times (2N+1)(K+1)}$, and the stated condition 
$2P\geq~(2N+~1)(K+1)$ is necessary for $L_1$ to have full column rank.
Because $\rank(L_1L_1^T) = \rank(L_1)$, we consider $L_1L_1^T$, which is more convenient to analyze.
\begin{equation}
\begin{aligned}
    L_1 L_1^T &= (\widehat{\Theta} \bullet \widehat{\Psi})(\widehat{\Theta} \bullet \widehat{\Psi})^T \\
    &= (\widehat{\Theta} \bullet \widehat{\Psi})(\widehat{\Theta}^T \star \widehat{\Psi}^T)
\end{aligned}
\end{equation}
Using the mixed product property $(A \bullet B)(C \star D) = (AC) \odot (BD)$ \cite{710918}, 
\begin{equation}
    \begin{aligned}
        L_1 L_1^T &= \widehat{\Theta}\widehat{\Theta}^T \odot \widehat{\Psi}\widehat{\Psi}^T \\
        &= R \odot \tilde{R}
    \end{aligned}
\end{equation}
where $R = \widehat{\Theta}\widehat{\Theta}^T$ and $\tilde{R} = \widehat{\Psi}\widehat{\Psi}^T$. 

Thanks to the assumption that the $P$ view angles $\theta_i \in [0, \pi], i=1, \ldots, P$ are distinct, it follows that the $2P$ view angles $\theta_i, \theta_i+\pi, i=1, \ldots, P$ are all distinct modulo $2\pi$, and thus the $2P$ exponentials  $e^{j\theta_i}, e^{j(\theta_i+\pi)}, i=1, \ldots, P $ defining the rows of $\widehat{\Theta}^T = [\Theta^T, \,\, \bar{\Theta}^T]$ are all distinct. Now, up to scaling by a full-rank diagonal matrix, matrix $\widehat{\Theta}^T$ is a Vandermonde matrix with distinct bases, and therefore has full Kruskal rank \cite{alexeev2012full}.
Thus, $\mathrm{krank}(\widehat{\Theta}) = 2N+1$.

This implies $\widehat{\Theta}$ has full column rank, hence $\rank(R) = \rank(\widehat{\Theta})= 2N+1$.
Let the eigendecomposition of $R \succcurlyeq 0$ be
\begin{equation}
    \begin{aligned}
        R &= V \Lambda V^T = \sum_{k=1}^{2N+1} \lambda_kv_kv_k^T
    \end{aligned}
\end{equation}
Note that under the assumption on $\Psi$ in Theorem~\ref{thm:L1_full_rank} we have
$ (1/\sqrt{2})\widehat{\Psi}^T (1/\sqrt{2})\widehat{\Psi} = I$.  It therefore follows that the eigendecomposition of $\tilde{R}\succcurlyeq 0$ is
\begin{equation}
    \tilde{R} =  \widehat{\Psi} (0.5I) \widehat{\Psi}^T
\end{equation}
Then, using the eigendecomposition property of the Hadamard product yields
\begin{equation}
    \begin{aligned}
 L_1 L_1^T =       R \odot \tilde{R} = 0.5\sum_{k=1}^{2N+1} \sum_{l=1}^{K+1} \lambda_k 
        (v_k \odot \widehat{\psi}_l)(v_k \odot \widehat{\psi}_l)^T
    \end{aligned}
\end{equation}
which implies that $\rank (L_1)$ is equal to the number of $v_k \odot \widehat{\psi}_l$ that are linearly independent, upper bounded by $(2N+1)(K+1) $. Therefore, we have the following

\begin{equation} \label{eq:rankL1}
    \begin{aligned} 
  \rank (L_1)      &= 
        \mathrm{rank} \left[ \mathrm{diag}(\widehat{\psi}_1)V 
        \ldots 
        \mathrm{diag}(\widehat{\psi}_{K+1})V \right] \\
        &\leq (2N + 1)(K + 1),
    \end{aligned}
\end{equation}
where the matrix on the right hand side has dimensions  ${2P \times (2N + 1)(K + 1)}$. 

To apply Corollary~\ref{cor:gauss_to_stiefel} to \eqref{eq:rankL1}, we require $\krank(V^T)$. Recall that $V$ contains the left singular vectors of $\widehat{\Theta}$. Consider the SVD $\widehat{\Theta}^T = U\Sigma V^T$.  Because $U$ and $\Sigma$ are invertible, we have $\krank(V^T)= \krank (\widehat{\Theta}^T) = 2N+1$. Then, since $\Psi$ is drawn randomly from the Stiefel manifold $V_{(K+1)}\left( \mathbb{R}^{P} \right)$ we may apply 
Corollary~\ref{cor:gauss_to_stiefel} with $n=2N+1$ and $L=K+1$ to \eqref{eq:rankL1}, to conclude that $\rank(L_1) = (2N+1)(K+1)$ is satisfied w.p. 1, completing the proof of Theorem \ref{thm:L1_full_rank}.
\qedsymbol

\subsection{Proof of Theorem~\ref{thm:L2_cond_num}}
To study the rank and condition number of $\bold{L}_2(\boldsymbol{\beta})$  we consider the Gram matrix $\mathbf{Y} \triangleq \bold{L}_2(\boldsymbol{\beta})^T\bold{L}_2(\boldsymbol{\beta})$ and then use the facts that $\mathrm{rank}(A) = \mathrm{rank}(A^TA), \,\, \forall A \in \mathbb{C}^{m \times n}$, and $\kappa(A^TA) =\kappa^2(A)$.

Using \ref{eq:L_2_ith_row_def}, we can write $ L_2(\boldsymbol{\beta})^{(i)T}L_2(\boldsymbol{\beta})^{(i)}$ as

\begin{equation}
    L_2(\boldsymbol{\beta})^{(i)T}L_2(\boldsymbol{\beta})^{(i)}
    = J\tilde{U}_i^TA_i \Gamma A_i^T\tilde{U}_i.
\end{equation}
where $\tilde{U}_i = I_{K+1} \otimes \widehat{U}_{i:}$ and $A_i$ is defined in \eqref{eq:A_i_def}.

Let $\lambda_1 \geq \ldots \geq \lambda_{d(K+1)}$ be the eigenvalues of $\Gamma$. Since $\lambda_1 I \succcurlyeq \Gamma \succcurlyeq \lambda_{d(K+1)} I$, we have 
\begin{align}
    \label{eq:L2_ineq_v1}
    \lambda_{d(K+1)}\tilde{U}_i^T \tilde{U}_i \preccurlyeq \frac{L_2(\boldsymbol{\beta})^{(i)T}L_2(\boldsymbol{\beta})^{(i)}}{J(2N+1)}
    &\preccurlyeq \lambda_1 \tilde{U}_i^T \tilde{U}_i,
\end{align}
where $A_i A_i^T = ||\widehat{\Theta}_{i:}||_2^2 I_{K+1}$ using the mixed product identity, i.e. $(A \otimes B)(C \otimes D) = (AC) \otimes (BD)$, and for all $i$, $||\widehat{\Theta}_{i:}||_2^2 = 2N+1$.

Then, using the mixed product identity again, 
\begin{equation}
    \begin{aligned}
        \label{eq:U_identity}
        \sum_{i=1}^{2P} \tilde{U}_i^T \tilde{U}_i &= \sum_{i=1}^{2P} I_{K+1} \otimes \widehat{U}_{i:}^T\widehat{U}_{i:}. \\
        &= I_{K+1} \otimes (\widehat{U}^T\widehat{U}) \\
        &= I_{K+1} \otimes I_d = I_{d(K+1)}.
    \end{aligned}
\end{equation}
and combining \eqref{eq:L2_ineq_v1} and \eqref{eq:U_identity} yields 
\begin{equation} \label{eq:Ybounds}
\begin{aligned}
    \lambda_{d(K+1)} I_{d(K+1)} \preccurlyeq \frac{\mathbf{Y}}{ J(2N+1)}
    \preccurlyeq \lambda_1 I_{d(K+1)}.
\end{aligned}
\end{equation}
This result indicates that the condition number of $\mathbf{Y}$ is bounded by $\lambda_1 / \lambda_{d(K+1)}$, i.e.
\begin{equation}
\kappa(\mathbf{Y}) \leq \kappa(\Gamma) = \lambda_1 / \lambda_{d(K+1)}.
\end{equation}
It follows that 
\begin{equation} \label{eq:KappaLtwo-bound}
    \kappa(\bold{L}_2(\boldsymbol{\beta})) \leq \sqrt{\kappa(\Gamma)},
\end{equation} 
establishing the upper bound on the condition number of $\bold{L}_2(\boldsymbol{\beta})$.

Next we state that based on the assumption $\Gamma \succ 0$, we have $\lambda_{d(K+1)}>0$, and $\mathbf{Y}$ and $\bold{L}_2(\boldsymbol{\beta})$ are full column rank.  Since $\Gamma \in \mathbb{R}^{(2N+1)(K+1) \times (2N+1)(K+1)}$, a necessary condition for $\Gamma \succ 0$ is $J \geq (2N+1)(K+1)$. A sufficient condition is that there are at least $(2N+1)(K+1)$ linearly independent vectors in the set $\{\beta(s_j), j = 1, \ldots, J \}$. 
\qedsymbol

\end{document}